\documentclass[submission]{eptcs}
\usepackage[utf8]{inputenc}
\usepackage{amssymb,mathrsfs,amsmath,amsfonts,amsthm,graphicx,pdfpages,todonotes,bbold}
\usepackage{tikz}

\def\titlerunning{Compositional Modeling with Stock-Flow Diagrams}
\def\authorrunning{Baez, Li, Libkind, Osgood and Patterson}

\usepackage[frozencache]{minted}

\usepackage{quiver}

\usepackage{xcolor}
\usepackage{framed,color}

\definecolor{shadecolor}{rgb}{1,0.8,0.3}
\definecolor{myurlcolor}{rgb}{0.5,0,0}
\definecolor{mycitecolor}{rgb}{0,0,0.7}
\definecolor{myrefcolor}{rgb}{0,0,0.7}
\definecolor{hyperrefcolor}{rgb}{0.5,0,0}

\usepackage{hyperref}
\hypersetup{
	colorlinks,
	linkcolor={hyperrefcolor},
	citecolor={mycitecolor},
	urlcolor={hyperrefcolor}
}

\newcommand{\define}[1]{{\bf \boldmath{#1}}}

\newcommand{\namedset}[1]{\mathbb{#1}}

\newcommand{\R}{\namedset{R}}

\newcommand{\namedcat}[1]{\mathsf{#1}}
\newcommand{\Set}{\namedcat{Set}}
\newcommand{\Fin}{\namedcat{Fin}}

\newcommand{\C}{\namedcat{C}}
\newcommand{\G}{\namedcat{G}}
\renewcommand{\H}{\namedcat{H}}

\newcommand{\Cat}{\namedcat{Cat}}

\newcommand{\Dynam}{\namedcat{Dynam}}
\newcommand{\StockFlow}{\namedcat{StockFlow}}
\newcommand{\Open}{\namedcat{Open}}

\newcommand{\X}{\namedcat{X}}
\newcommand{\A}{\namedcat{A}}

\newcommand{\ff}{f}

\newcommand{\stock}{\mathrm{stock}}
\newcommand{\flow}{\mathrm{flow}}
\newcommand{\link}{\mathrm{link}}
\newcommand{\variable}{\mathrm{variable}}

\newcommand{\ICU}{\mathrm{ICU}}
\newcommand{\NICU}{\mathrm{NICU}}

\makeatletter
\newcommand*{\relrelbarsep}{.386ex}
\newcommand*{\relrelbar}{%
  \mathrel{%
    \mathpalette\@relrelbar\relrelbarsep
  }%
}
\newcommand*{\@relrelbar}[2]{%
  \raise#2\hbox to 0pt{$\m@th#1\relbar$\hss}%
  \lower#2\hbox{$\m@th#1\relbar$}%
}
\providecommand*{\rightrightarrowsfill@}{%
  \arrowfill@\relrelbar\relrelbar\rightrightarrows
}
\providecommand*{\leftleftarrowsfill@}{%
  \arrowfill@\leftleftarrows\relrelbar\relrelbar
}
\providecommand*{\xrightrightarrows}[2][]{%
  \ext@arrow 0359\rightrightarrowsfill@{#1}{#2}%
}
\providecommand*{\xleftleftarrows}[2][]{%
  \ext@arrow 3095\leftleftarrowsfill@{#1}{#2}%
}
\makeatother

\usepackage{mathtools}
\usepackage[utf8]{inputenc}
\usepackage{csquotes}
\usepackage{color}
\usepackage{tikz}

\usepackage{comment}

\usepackage{graphicx}
\usepackage{adjustbox}
\usepackage[all,2cell]{xy}\UseAllTwocells\SilentMatrices
\definecolor{darkgreen}{rgb}{0,0.45,0}

\usepackage[capitalize]{cleveref}
\crefname{equation}{}{}
\crefname{item}{}{}
\usepackage{enumerate}

\newtheorem*{thm*}{Theorem}
\theoremstyle{remark}
\newtheorem*{rmk*}{Remark}
\newtheorem*{lem*}{Lemma}
\theoremstyle{definition}
\newtheorem*{defn*}{Definition}
\newtheorem*{cor*}{Corollary}
\theoremstyle{definition}
\newtheorem*{examples*}{Examples}
\newtheorem{prop*}{Proposition}

\theoremstyle{plain}
\newtheorem{thm}{Theorem}[section]
\theoremstyle{plain}

\theoremstyle{remark}

\theoremstyle{plain}

\theoremstyle{plain}

\theoremstyle{definition}

\theoremstyle{definition}

\newcommand{\maps}{\colon}

\usepackage{tikz}
\usepackage{tikz-cd}
\usetikzlibrary{backgrounds,circuits,circuits.ee.IEC,shapes,fit,matrix}

\tikzstyle{simple}=[-,line width=2.000]
\tikzstyle{arrow}=[-,postaction={decorate},decoration={markings,mark=at position .5 with {\arrow{>}}},line width=1.100]
\pgfdeclarelayer{edgelayer}
\pgfdeclarelayer{nodelayer}
\pgfsetlayers{edgelayer,nodelayer,main}

\tikzstyle{none}=[inner sep=0pt]

\definecolor{lblue}{rgb}{0,250,255}
\tikzstyle{species}=[circle,fill=yellow,draw=black,scale=1.15]
\tikzstyle{transition}=[rectangle,fill=lblue,draw=black,scale=1.15]
\tikzstyle{inarrow}=[->, >=stealth, shorten >=.03cm,line width=1.5]
\tikzstyle{empty}=[circle,fill=none, draw=none]
\tikzstyle{inputdot}=[circle,fill=black,draw=black, scale=.25]
\tikzstyle{inputarrow}=[->,draw=purple, shorten >=.05cm]
\tikzstyle{simple}=[-,draw=black,line width=1.000]

\usetikzlibrary{arrows,shapes,automata,backgrounds,petri}
\tikzstyle{place}=[circle,thick,draw=blue!75,fill=blue!20,minimum size=6mm]
\tikzstyle{red place}=[place,draw=red!75,fill=red!20]
\tikzstyle{transition}=[rectangle,thick,draw=black!75,
  			  fill=black!20,minimum size=4mm]

\tikzset{-|->/.style={decoration={markings,
      mark=at position 0.5 with {\arrow{|}},
      mark= at position 1 with{\arrow{>}}},
    postaction={decorate}}}

\providecommand{\event}{ACT 2022}

\title{Compositional Modeling with Stock and Flow Diagrams}
\author{John Baez
\institute{Department of Mathematics \\ U.\ C.\ Riverside \\ California, USA}
\email{baez@math.ucr.edu}
\and
Xiaoyan Li
\institute{Department of Computer Science \\ University of Saskatchewan \\ Saskatoon, Canada}
\email{xiaoyan.li@usask.ca}
\and
Sophie Libkind
\institute{Department of Mathematics  \\ Stanford University \\ University of California, USA}
\email{slibkind@stanford.edu}
\and 
Nathaniel D. Osgood
\institute{Department of Computer Science\\ University of Saskatchewan \\ Saskatoon, Canada}
\email{nathaniel.osgood@usask.ca}
\and
Evan Patterson
\institute{Topos Institute \\ California, USA}
\email{evan@topos.institute}
}

\date{July 11, 2022}

\begin{document}

\maketitle

\begin{abstract}
Stock and flow diagrams are widely used in epidemiology to model the dynamics of populations.  Although tools already exist for building these diagrams and simulating the systems they describe, we have created a new package called StockFlow, part of the AlgebraicJulia ecosystem, which uses ideas from category theory to overcome notable limitations of existing software.  Compositionality is provided by the theory of decorated cospans: stock and flow diagrams can be composed to form larger ones in an intuitive way formalized by the operad of undirected wiring diagrams.  Our approach also cleanly separates the syntax of stock and flow diagrams from the semantics they can be assigned.  We consider semantics in ordinary differential equations, although others are possible.  As an example, we explain code in StockFlow that implements a simplified version of a COVID-19 model used in Canada.
\end{abstract}
\section{Introduction}

The theoretical advantages of compositionality and functorial semantics are widely recognized among applied category theorists. \textit{Compositionality} means, at the very least, that systems can be described one piece at a time, with a clear formalism for composing these pieces.  This formalism can appear in various styles: composing morphisms in a category, tensoring objects in a monoidal category, composing operations in an operad, etc.  \emph{Functorial semantics} then means that the map from system descriptions (``syntax'') to their behavior (``semantics'') preserves all the relevant forms of composition.

While these principles are elegant, in many fields it is still a challenge to produce useful software that takes advantage of them and is embraced by the intended users.  This is one of the main challenges of applied category theory.   Here we focus on developing software suited to one particular field: epidemiological modeling.  At present this software is additionally capable of modeling a wide class of systems studied in the System Dynamics modeling discipline \cite{forrester,sterman2000business}.

The AlgebraicJulia ecosystem of software implements compositionality and functorial semantics in a thorough-going way \cite{AlgebraicJulia}.   Decorated and structured cospans  are broad mathematical frameworks for turning ``closed'' system descriptions into ``open'' ones that can be composed along their boundaries \cite{fong2015,baezcourser2020,baez-courser-vasilakopoulou2022}.  One part of AlgebraicJulia, called Catlab \cite{patterson-lynch-fairbanks2021}, provides a generic interface for working with such cospans, among other categorical abstractions.   With the help of Catlab, a tool called AlgebraicPetri was developed to work with one approach to epidemiological modeling based on Petri nets \cite{libkind-baas-halter-patterson-fairbanks2022}.  Here we explain a new tool, StockFlow, which handles a more flexible and more widely used formalism for epidemiological modeling: stock and flow diagrams.

In \cref{sec:epidemiology} we review how stock and flow diagrams are used in epidemiological modeling, and discuss some shortcomings of existing software for working with these diagrams.  In \cref{sec:math} we first use decorated cospans to formalize a simple class of open stock and flow diagrams and their differential equation semantics, and then sketch how to extend this class to the full-fledged diagrams actually used in our software.   In \cref{sec:implementation} we describe the software package, StockFlow, that we have developed to work with stock-flow diagrams compositionally and implement a functorial semantics for them.  The reader can find the StockFlow repository on GitHub at \url{https://github.com/AlgebraicJulia/StockFlow.jl}.

\section{Epidemiological modeling with stock and flow diagrams}
\label{sec:epidemiology}

Effective decision-making regarding prevention, control, and service delivery to address the health needs of the population involves reasoning about diverse complexities: policy resistance, feedbacks, heterogeneities, multi-condition interactions, and nonlinearities that collectively give rise to counterintuitive results \cite{sterman1994learning,trecker2015revised}.  For over a century, researchers and practitioners have used epidemiology models to address such challenges.  Since dynamic epidemiological modeling was first applied to communicable diseases \cite{kermack1927contribution,ross1916application, ross1917application}, it has both deepened its reach in that area \cite{anderson1992infectious, diekmann2000mathematical} and spread to many other subdomains of epidemiology, including chronic, behavioural, environmental, occupational, and social epidemiology, as well as spheres such as mental health and addictions.  Reflecting a world in which growing global interconnection is juxtaposed with increasing ecosystem encroachment and climate stresses, the rise of the ``One Health" perspective \cite{destoumieux2018one,mackenzie2019one} has motivated such modeling to increasingly incorporate dynamics from domains such as ecology, veterinary and agricultural health, and social dynamics and inequities.  Such efforts have come to define the field of mathematical and computational epidemiology.

The earliest and still most common epidemiological models are sets of ordinary differential equations, typically used to characterize epidemiological dynamics in an aggregate fashion \cite{anderson1992infectious}. Delay and partial differential equations have also been widely applied. Recent decades have witnessed a rapid growth in use of agent-based models.
Although the techniques explored here may be more widely applicable, we focus on aggregate models described using differential equations.


Contemporary aggregate-level modeling involves widespread informal use of diagrams, with the most prevalent type  of such diagrams 
being transition diagrams and their richer and more formal cousins, ``stock and flow diagrams'' \cite{sterman2000business}, as depicted in Figure \ref{fig:ExStockAndFlow}. Transition diagrams are a minimalist box-and-arrow formalism which draws state variables as boxes and transitions as arrows. Traditionally, most mathematical epidemiologists have focused directly on the underlying differential equations, regarding such diagrams only as an informal presentation of the equations.  Thus, diagrams are commonly treated either as ephemeral artifacts useful for thinking out structures and then discarded, or as an expedient aid for communication.

\begin{figure*}
    \centering
    \includegraphics[width=.9\linewidth]{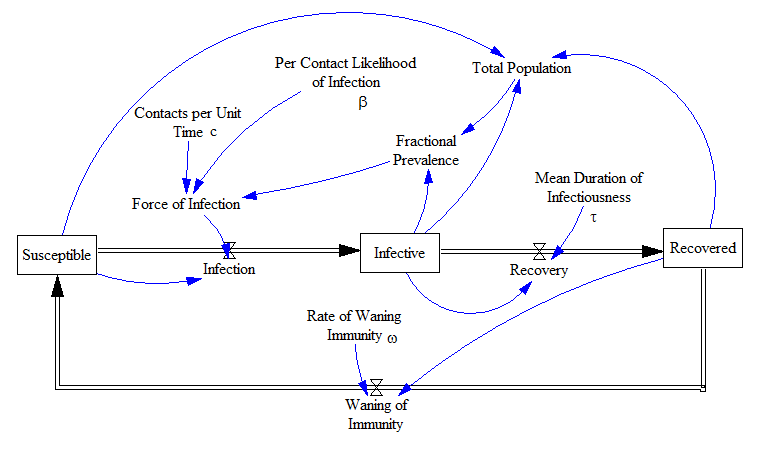}
    \caption{An example stock-flow diagram. This stock-flow diagram has three stocks labeled ``Susceptible", ``Infective", and ``Recovered"; three flows labeled ``Infection", ``Recovery", and ``Waning of Immunity"; many auxiliary variables including ``Force of Infection" and ``Total Population"; and links depicted by blue arrows. }
    \label{fig:ExStockAndFlow}
    \label{fig:ExShttps://www.overleaf.com/project/622660b1168ab0417506df38tockAndFlow}
\end{figure*}

Amongst the notable minority of health modelers who employ stock and flow diagrams (also called ``Forrester diagrams'', and termed here ``stock-flow diagrams" for brevity), these diagrams play roles at different stages of the modeling process. Stock-flow diagrams depict state variables as stocks (rectangles), changes to those stocks as flows (thick arrows also termed ``material connections"), constants and auxiliary variables (also called ``dynamic variables"), and links (arrows sometimes called ``informational connections" or simply ``connections") characterizing instantaneous dependencies. 

Stock-flow diagrams serve as the central formalism in the modeling tradition of System Dynamics, initiated by Forrester in the 1950s \cite{forrester, sterman2000business}.  Since the 1980s, System Dynamics software has provided refined, visually accessible, declarative user interfaces for interactively building, and browsing stock-flow diagrams \cite{richmond1985stella}.  While such packages are designed to ensure transparency of model structure to modelers and stakeholders \cite{richmond1985stella}, they also serve as simulation tools. For that purpose, the System Dynamics tradition  universally interprets stock-flow diagrams as characterizing ordinary differential equations.  Stocks represent the state variables; their formulation requires specifying an initial value. Flows represent the differentials associated with stocks, and are each associated with a modeler-specified mathematical expression specifying the flow rate (quantity per unit time) as a function of other variables. Each constant variable is associated with a real scalar. Auxiliary variables generally reflect quantities of domain significance that depend instantaneously on other model quantities. Each such auxiliary variable is associated with a modeler-specified expression characterizing the value of that auxiliary as a function of the current value of other variables (stocks, flows, constants and other auxiliary variables).

Reflecting the strong emphasis that System Dynamics practice places on stakeholder engagement and participatory model building, models built in the stock-flow paradigm are routinely shown to stakeholders without modeling background---be they domain experts from a modeling team, stakeholders, or community members---to elicit critiques and suggestions \cite{CBSD,richmond1985stella,vennix1996group}. System Dynamics has also long sought to recognize, codify, and exploit widespread use of modeling idioms. Thus, researchers and practitioners have formalized dozens of simple stock-flow diagrams called ``molecules'' for reuse in modeling \cite{molecules}. Some simulation packages provide molecules as pre-specified templates defined by the software, and mechanisms for for directly incorporating built-in templates for such molecules into models.  When a molecule is added to a model, the elements of the molecule---such as stocks and flows---are simply added as elements of the surrounding diagram, rather than being reused as higher level abstractions. Moreover, because such libraries of molecules are fixed, such molecules cannot be created or packaged up by the user.  

Software packages for stock-flow diagrams have as a central feature the simulation, via numerical integration, of the system of ordinary differential equations described by these diagrams.  Many such packages also offer additional forms of model analysis, including identification of feedback loops, performing tests of dimensional homogeneity based on modeler unit annotations, sensitivity analysis and calibration. Some tools support more sophisticated forms of analysis, such as those involving Markov Chain Monte Carlo and extended Kalman filtering.  
But while existing stock-flow modeling tools offer refined interfaces for building, exploring and simulating models, their support for modern modeling practice is hampered by significant rigidity and several additional shortcomings.  The present paper focuses on addressing two limitations of contemporary tools.
  
First, and most notable from a categorical perspective, existing tools \emph{lack support for composition} of models, despite there being several natural ways in which models might be composed.  Instead, each model is currently treated in isolation.  If models are composed at all, it is by either outputting data files from one and importing such data into another, or by creating, via an ad hoc process, a third model that contains both of the original models.

Second, existing stock-flow modeling tools \emph{privilege a single semantics} associated with stock-flow diagrams: the interpretation of these diagrams as ordinary differential equations. While alternative interpretations can sometimes be force-fit---for example, a difference equation interpretation by using Euler integration, or a stochastic differential equation interpretation using formulas for flows drawing from suitable probability distributions---they are commonly awkward, obscure, and error-prone.  Although particular packages allow for select additional analyses---for example, identification of feedback loops---such features are hard-coded, and many analysis tools demonstrated as valuable by research \cite{ guneralp2005progress,kampmann2012feedback,saleh2005comprehensive} have not been incorporated in extant software packages.

We turn next to a mathematical framework that provides a remedy for these deficiencies: an explicitly compositional framework where ``open'' stock-flow diagrams become morphisms in a category and where semantics is described as a functor from this category to some other category.  For reasons of space we only describe one choice of semantics, but the clear separation of syntax and semantics permits swapping out this choice for others.

\section{The mathematics of stock-flow diagrams}
\label{sec:math}

Stock-flow diagrams come in many variants.  To illustrate our methodology we begin with a very simple kind.  In our code we have implemented a more sophisticated variant with additional features, but the ideas are easier to explain without those features.   Our main goal is to study \emph{open} stock-flow diagrams---that is, stock-flow diagrams in which various stocks are specified as ``interfaces.''   We can treat open stock-flow diagrams with two interfaces as morphisms of a category. Composing these morphisms then lets us build larger diagrams from smaller ones.  Alternatively, we can compose stock-flow diagrams with any number of interfaces using an operad.   Both approaches let us describe the differential equation semantics for open stock-flow diagrams following a paradigm already explored for open Petri nets with rates \cite{baez-courser-vasilakopoulou2022,baezpollard2017}.  We describe that paradigm here.


\subsection{A category of stock-flow diagrams}
\label{subsec:category}

As a first step, we define ``primitive'' stock-flow diagrams with stocks, flows, and links but not the all-important functions that describe the rate of each flow.  For this, we consider a category $\H$ freely generated by these objects and morphisms:
    \[
    \begin{tikzcd}
        \flow
        \arrow[rr, shift left = 1, "u"]
        \arrow[rr, shift right = 1, "d", swap]
        & &
        \stock \\
        & \link 
        \arrow[ul, "t"] \arrow[ur, "s", swap]
    \end{tikzcd}
    \]
We call a functor $F \maps \H \to \Fin\Set$ a \define{primitive stock-flow diagram}.  It amounts to the following:
\begin{itemize}
\item a finite set of stocks $F(\stock)$,
\item a finite set of flows $F(\flow)$, 
\item functions $F(u), F(d) \maps F(\flow) \to F(\stock)$ assigning to each flow the stock \define{upstream} from it, and the stock \define{downstream} from it,
\item a finite set of links $F(\link)$,
\item functions $F(s) \maps F(\link) \to
F(\stock), F(t) \maps F(\link) \to F(\flow)$ assigning to each link its \define{source}, which is a stock, and its \define{target}, which is a flow.
\end{itemize}
Given $f \in F(\flow)$, we say $f$ \define{flows from} the upstream stock $F(u)(f)$ and \define{flows to} the downstream stock $F(d)(f)$.  We say that a link $\ell \in F(\link)$ \define{points from} its source $F(s)(\ell)$ and \define{points to} its target $F(t)(\ell)$.    There is a category of primitive stock-flow diagrams, $\Fin\Set^{\H}$, where the objects are functors from $\H$ to $\Fin\Set$ and a morphism from $F \maps \H \to \Fin\Set$ to $G \maps \H \to \Fin\Set$ is a natural transformation.

Primitive stock-flow diagrams are useful for \emph{qualitative} aspects of modeling, since they clearly show which flows depend on which stocks; as such, they can be seen as a restricted form of \textit{system structure diagrams} \cite{lich2020engaging, faghihi2015sustainable} used in System Dynamics practice.   But they become useful for \emph{quantitative} modeling and simulation only when we equip them with functions saying how the rate of each flow depends on the value of each stock.  Thus, we define a \define{stock-flow diagram} to be a pair $(F,\phi)$ consisting of an object $F \in \Fin\Set^{\H}$ and a continuous function called a \define{flow function}
\[  \phi_f \maps \R^{F(t)^{-1}(f)} \to \R \]
for each flow $f \in F(\flow)$, where $F(t)^{-1}(f)$ is the set of links with target $f$:
\[  F(t)^{-1}(f) = \{ \ell \in F(\link)  \; \vert \; F(t)(\ell) = f \} .\]
The idea is that the flow function $\phi_f$ says how the rate of the flow $f$ depends on the values of all the stocks with links pointing to it.  We make this precise in \cref{subsec:semantics} when we introduce a semantics that maps each stock-flow diagram to a first-order differential equation.  But rates and values play no formal role in this section.

To define a category of stock-flow diagrams, we need to define morphisms between them.  What is a morphism from $(F,\phi)$ to $(G,\psi)$?   It is a natural transformation $\alpha \colon F \Rightarrow G$ with an extra property.  Because $\alpha$ is natural, we get a commutative square
 \[\begin{tikzcd}
        F(\link) 
        \arrow[r, "\alpha(\link)"]
        \arrow[d, "F(t)", swap]
        & G(\link) 
        \arrow[d, "G(t)"] \\
        F(\flow) 
        \arrow[r, "\alpha(\flow)", swap]
        & G(\flow).
    \end{tikzcd}\]
Thus, letting $g = \alpha(\flow)(f)$ for $f \in F(\flow)$, we get a map
\[ \alpha(\link) \maps F(t)^{-1} (f) \to
G(t)^{-1} (g) \]
and thus a linear map
\[   \alpha(\link)^\ast \maps \R^{G(t)^{-1} (g)} \to \R^{F(t)^{-1} (f)} \]
given by precomposition:
\[   \alpha(\link)^\ast (x) = x \circ \alpha(\link) .\]
We say that $\alpha$ is a \define{morphism of stock-flow diagrams} from $(F,\phi)$ to $(G,\psi)$ if
\begin{equation}\label{eq:pushforward}
    \psi_g = \sum_{f \in \alpha(\text{flow})^{-1}(g)} \phi_f \circ \alpha(\text{link})^*
\end{equation}
for every $g \in G(\flow)$.    This equation expresses rates of flows in $(G,\psi)$ as sums of rates of flows in $(F,\phi)$.   For example, Figure \ref{fig:SIR_to_SIS} shows a morphism of stock-flow diagrams in which two flows, ``recovery" $r$ and ``death" $d$, are mapped to a single ``removal" flow $e$.   The above equation implies that
\[   \psi_e = \phi_r \circ \alpha(\text{link})^* + 
\phi_d \circ \alpha(\text{link})^* . \]
This equation says that the rate of the flow $e$ is the 
sum of the rates of $r$ and $d$.   

Composition of morphisms between stock-flow diagrams is just composition of their underlying natural transformations; one can show that indeed the composite of two natural transformations obeying Eq.\ \cref{eq:pushforward} again obeys this equation.  We thus obtain a category of stock-flow diagrams, which we call $\StockFlow$.

\begin{figure}
    \centering
    \includegraphics[width=0.9\linewidth]{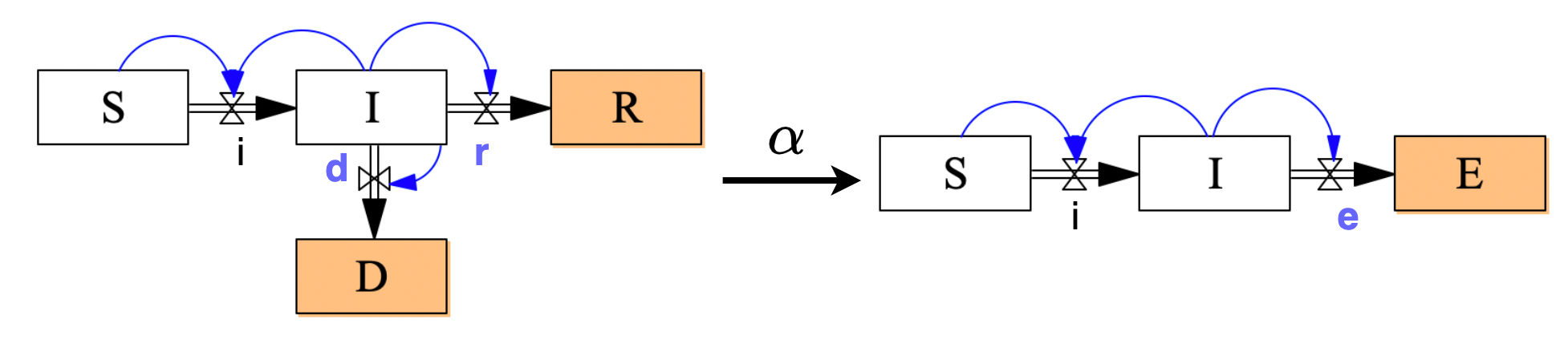}
    \caption{At left is a stock-flow diagram $(F,\phi)$ with stocks $S, I, R, D$ corresponding to  susceptible, infected, recovered and deceased  populations and flows $i,r,d$ corresponding to infection, recovery and death.  Links are shown in blue.  At right we see a simpler stock-flow diagram $(G,\psi)$ where recovered and deceased populations are lumped into a single ``removed'' stock $E$, and recovery and death are lumped into a single ``removal'' flow $e$.  There is an evident morphism $\alpha \maps (F,\phi) \to (G,\psi)$ sending $R$ and $D$ to $E$ and sending $r$ and $d$ to $e$.}   
    \label{fig:SIR_to_SIS}
\end{figure}

\subsection{Open stock-flow diagrams}
\label{subsec:open}

We can build larger stock-flow diagrams by gluing together smaller ones.   There are a number of choices of how to formalize this.   Here we glue together two stock-flow diagrams by identifying two collections of stocks to serve as ``interfaces."  Thus, we define an \define{open stock-flow diagram} with finite sets $X$ and $Y$ as interfaces to be a stock-flow diagram $(F,\phi)$ equipped with functions from $X$ and $Y$ to its set of stocks:
\[
\begin{tikzcd}
& F(\stock) & \\
X\ar[ur,"i"] && Y \ar[ul,swap,"o"]
\end{tikzcd}
\]
We call this an open stock-flow diagram from $X$ to $Y$ and write it tersely as $(F,\phi) \colon X \to Y$, despite the maps $i$ and $o$ being a crucial part of the structure.

We can compose open stock-flow diagrams from $X$ to $Y$ and from $Y$ to $Z$ to obtain one from $X$ to $Z$.  To formalize this composition process we use Fong's theory of decorated cospans \cite{fong2015}. However, to make composition associative and get a category we need to use \emph{isomorphism classes} of open stock-flow diagrams.   Two open stock-flow diagrams $(F,\phi)$ and $(F',\phi')$ from $X$ to $Y$ are \define{isomorphic} if there is an isomorphism of stock-flow diagrams $\alpha \maps (F,\phi) \to (F',\phi')$ such that this diagram commutes:
\[
\begin{tikzcd}
& F(\stock)\ar[dd,swap,"\alpha(\stock)"] & \\
X\ar[ur,"i"]\ar[dr,swap,"i'"] && Y \ar[ul,swap,"o"] \ar[dl,"o'"]\\
& F'(\stock)   &
\end{tikzcd}
\]

Using the theory of decorated cospans, we obtain:

\begin{thm}
\label{thm:openstockflow}
There is a category $\Open(\StockFlow)$ such that:
\begin{itemize}
\item An object is a finite set $X$.
\item A morphism from $X$ to $Y$ is an isomorphism class of open stock-flow diagrams from $X$ to $Y$.
\end{itemize}
This is a symmetric monoidal category, indeed a hypergraph category.
\end{thm}

\begin{proof}
This follows from the theory of decorated cospans \cite{fong2015} once we check the following facts.  For each finite set there is a category $\mathsf{C}(S)$ whose
\begin{itemize} 
\item objects are stock-flow diagrams with $S$ as their set of stocks, and
\item morphisms are morphisms $\alpha$ of stock-flow diagrams where $\alpha(\text{stock})$ is the identity on $S$.
\end{itemize}
Let $C(S)$ be the set of isomorphism classes of objects in this category.  A map of finite sets $f \maps S \to S'$ functorially determines a map from $C(S)$ to $C(S')$, and the resulting functor $C$ is symmetric lax monoidal from $(\Fin\Set, +)$ to $(\Set, \times)$, where the laxator
\[   \gamma \maps C(S) \times C(S') \to C(S + S')  \]
maps a pair of stock-flow diagrams to their ``disjoint union."  It follows from  \cite[Prop.\ 3.2]{fong2015} that we get the desired symmetric monoidal category $\Open(\StockFlow)$, and from \cite[Thm.\ 3.4]{fong2015}  that this is a hypergraph category.
\end{proof}

The point of making open stock-flow diagrams into the morphisms of a hypergraph category is that it gives ways of composing these diagrams that are more flexible than just composing them ``end-to-end" (ordinary composition of morphisms) and ``side-by-side" (a parallel arrangement expressed by tensoring).   Indeed, hypergraph categories are algebras of an operad, sometimes called the operad of undirected wiring diagrams, that encapsulates a wide range of composition strategies \cite{fongspivak2018}.  We use this approach in our code, and instead of working with cospans we actually use multicospans \cite{libkind-baas-patterson-fairbanks2021,spivak2013}, a mild generalization that allows for open stock-flow diagrams with any number of interfaces, not just two.

We conclude with a technical remark on Theorem \ref{thm:openstockflow}.  In fact there is a symmetric lax monoidal functor $\mathsf{C} \maps (\Fin\Set,+) \to (\Cat,\times)$ that sends each finite set $S$ to a \emph{category} of stock-flow diagrams with $S$ as their set of stocks.  Theorem 2.2 of  \cite{baez-courser-vasilakopoulou2022} thus gives a symmetric monoidal double category $\mathbb{O}\mathbf{pen}(\StockFlow)$ where objects are finite sets and horizontal 1-cells are actual open stock-flow diagrams, not mere isomorphism classes of these.  

This double category allows us to work with maps \emph{between} open stock-flow diagrams.  This should be useful for mapping several stocks to a single stock in a simplified stock-flow diagram, as in Figure 2, or embedding a stock-flow diagram in a more complicated one.  However, StockFlow currently does not attempt to support maps between open stock-flow diagrams, so \cref{thm:openstockflow} suffices for us.  Indeed, when working with a mere \emph{category} of open stock-flow diagrams, as opposed to a double category, we can define an isomorphic category using structured rather than decorated cospans: for open stock-flow diagrams, the difference only becomes visible at the double category level.  Thus, our treatment using decorated cospans looks forward to a future where we work with maps between open stock-flow diagrams.

\subsection{Open dynamical systems}
\label{subsec:opendynam}

Our next goal is to define a semantics for stock-flow diagrams mapping each such diagram to a dynamical system: a system of differential equations that describes the continuous-time evolution of the value of each stock.  This semantics is implicit in the usual applications of stock-flow diagrams; indeed, the stock-flow diagram is sometimes regarded merely as a convenient notation for a dynamical system.  While we illustrate the choice of a semantics for stock-flow diagrams using the continuous dynamical system interpretation, this semantics holds no privileged status, and there are several other semantics of practical value that could be employed instead.

In fact, our semantics is more general than suggested above: we describe a map from \emph{open} stock-flow diagrams to \emph{open} dynamical systems.  Our strategy for defining this semantics closely follows the strategy already used for open Petri nets with rates \cite{baezcourser2020,baez-courser-vasilakopoulou2022,baezpollard2017} and implemented for epidemiological models using AlgebraicJulia \cite{libkind-baas-halter-patterson-fairbanks2022}.  

For Petri nets with rates, the dynamics is typically described by the ``law of mass action," which only produces dynamical systems that are polynomial-coefficient vector fields on $\R^n$.   In stock-flow diagrams this restriction is dropped, but the rate of any flow out of its upstream stock equals the rate of flow into its downstream stock, so the total value of all stocks is conserved.  However, the more general stock-flow diagrams of \cref{subsec:full-schema} no longer obey this conservation law, since they allow ``inflows'' and ``outflows'' to the diagram as a whole.  With these generalizations, stock-flow diagrams become strictly more general than Petri nets with rates---at least in terms of the dynamical systems they can describe.

We begin by defining a \define{dynamical system} on a finite set $S$ to be a continuous vector field $v \maps \R^S \to \R^S$.  In our applications, $S$ will be the set of stocks of some stock-flow model, and the vector field $v$ is used to write down a differential equation describing the dynamics:
\[
\frac{d x(t)}{dt} = v(x(t)) 
\]
where at each time $t$, the vector $x(t) \in \R^S$ describes the value of each stock at time $t$.  Since the vector field is continuous, the Peano existence theorem implies that, for any initial value $x(0) \in \R^S$, the above equation has a solution for all $t$ in some interval $(-\epsilon, \epsilon)$.   However, the solution may not be unique unless we require that $v$ be nicer.  The theory we develop now can be modified to add extra restrictions, simply by replacing continuous functions with functions of a suitably nicer sort.

We define an \define{open} dynamical system from the finite set $X$ to the finite set $Y$ to be a pair $(S,v)$, consisting of a finite set $S$ and a dynamical system $v$ on $S$, together with functions from $X$ and $Y$ into $S$.  We depict this as follows:
\[
\begin{tikzcd}
& S &&  v \in D(S)\\[-10pt]
X\ar[ur,"i"] && Y \ar[ul,swap,"o"]
\end{tikzcd}  
\]
where $D(S)$ is the set of all dynamical systems on $S$.  Two open dynamical systems $(S,v)$ and $(S',v')$ from $X$ to $Y$ are \define{isomorphic} if there is a bijection $\beta \maps S \to S'$ such that the following diagram commutes:
\[
\begin{tikzcd}
& S \ar[dd,swap,"\beta"] && v \in D(S) \\[-10pt]
X\ar[ur,"i"]\ar[dr,swap,"i'"] && Y \ar[ul,swap,"o"] \ar[dl,"o'"]\\[-10pt]
& S' && v' \in D(S')
\end{tikzcd}
\]
and $\beta_* \circ v \circ \beta^* = v'$, where $\beta_*: \R^{S} \to \R^{S'}$ is the pushforward map defined by 

\[
\beta_*(x)(\sigma') = \sum_{\sigma \in \beta^{-1}(\sigma')} x(\sigma) \qquad \forall x \in \R^S,\ \sigma' \in S'.
\]

We can then construct a category where objects are finite sets and morphisms from $X$ to $Y$ are isomorphism classes of open dynamical systems from $X$ to $Y$.  This was done in \cite[Theorem 17]{baezpollard2017} by applying Fong's theory of decorated cospans to a functor $D \maps \Fin\Set \to \Set$ sending any finite set $S$ to the set $D(S)$ of  dynamical systems on $S$:

\begin{thm}[\textbf{Baez--Pollard}]
\label{thm:opendynam}
There is a category $\Open(\Dynam)$ such that:
\begin{itemize}
\item An object is a finite set $X$.
\item A morphism from $X$ to $Y$ is an isomorphism class of open dynamical systems from $X$ to $Y$.
\end{itemize}
This is a symmetric monoidal category, indeed a hypergraph category.
\end{thm}

In fact, there is a symmetric lax monoidal functor $\mathsf{D} \maps (\Fin\Set,+) \to (\Cat,\times)$ that maps any set $S$ to the discrete category on the set $D(S)$ described above.  The theory of decorated cospans then gives a symmetric monoidal \emph{double} category $\mathbb{O}\mathbf{pen}(\Dynam)$ where objects are finite sets and horizontal 1-cells are open dynamical systems.  This is discussed in \cite[Sec.\ 6.4]{baez-courser-vasilakopoulou2022}.

\subsection{Open dynamical systems from open stock-flow diagrams}
\label{subsec:semantics}

Next we describe a functor sending any open stock-flow diagram to an open dynamical system.    Suppose we have an open stock-flow diagram $(F,\phi) \maps X \to Y$, equipped with the cospan
\[
\begin{tikzcd}
& S & \\
X\ar[ur,"i"] && Y \ar[ul,swap,"o"]
\end{tikzcd}  
\]
where $S = F(\stock)$.  Then there is an open dynamical system $v(F,\phi)$ on $S$ given by
\begin{equation}
\label{eq:v1}
v(F,\phi)(x)(\sigma) = \sum_{f \in F(d)^{-1}(\sigma)} \phi_f(x \circ F(s)) \; - \sum_{f \in F(u)^{-1}(\sigma)} \phi_f(x \circ F(s)) 
\qquad \forall x \in \R^S, \sigma \in S .
\end{equation}
This formula looks a bit cryptic, so let us explain it.
Taking the expression $\R^S$ seriously, we can think of $x \in \R^S$ as a real-valued function on the set $S$ of stocks.  Each flow $f \in F(\flow)$ has a set $F(t)^{-1}(f)$ of links with $f$ as target, so there is an inclusion of sets $F(t)^{-1}(f) \hookrightarrow F(\link)$, and we can thus form the composite 
\[
F(t)^{-1}(f) \hookrightarrow F(\link) \xrightarrow{F(s)} F(\stock) \xrightarrow{x} \R
\] 
which for short we call simply
\[    x \circ F(s) \in \R^{F(t)^{-1}(f)}.\]
For each link $\ell$ with $f$ as its target, this composite gives the value of the stock that is $\ell$'s source.  Applying the function $\phi_f \maps \R^{F(t)^{-1}(f)} \to \R$, we obtain the rate of the flow $f$:
\[    \phi_f ( x \circ F(s)) \in \R .\]
This quantity has the effect of increasing the stock $d(f)$ and also decreasing the stock $u(f)$.   Thus the rate of change of any stock $\sigma \in S$ is
\[   \sum_{f \in F(d)^{-1}(\sigma)} \phi_f(x \circ F(s)) - \sum_{f \in F(u)^{-1}(\sigma)} \phi_f(x \circ F(s)). \]
This gives our formula for $v(F,\phi)$ in Equation \cref{eq:v1}.

Now, recall that in Theorem \ref{thm:openstockflow} the category of open stock-flow diagrams was defined as a decorated cospan category using the functor $C \maps \Fin\Set \to \Set$, while in Theorem \ref{thm:opendynam} the category of open dynamical systems was defined in a similar way using the functor $D \maps \Fin\Set \to \Set$.  According to the theory \cite{fong2015}, to obtain a semantics mapping open stock-flow diagrams to open dynamical systems, we need to define a natural transformation $\theta \maps C \Rightarrow D$. We do this as follows: for each finite set $S$, define $\theta(S)$ to map the isomorphism class $(F,\phi)$ in $C(S)$ to the isomorphism class of $v(F,\phi)$ in $D(S)$.   

\begin{thm}
\label{thm:semantics}
There is a functor
\[  v \maps \Open(\StockFlow) \to \Open(\Dynam) \]
sending
\begin{itemize}
\item any finite set to itself,
\item the isomorphism class of the stock-flow diagram $(F,\phi)$ made open as follows:
\[
\begin{tikzcd}
& F(\stock) & \\
X\ar[ur,"i"] && Y \ar[ul,swap,"o"]
\end{tikzcd}
\]
to the isomorphism class of the open dynamical system 
\[
\begin{tikzcd}
& F(\stock) &&  v(F,\phi) \in D(S).\\
X\ar[ur,"i"] && Y \ar[ul,swap,"o"]
\end{tikzcd}  
\]
\end{itemize}
This is a symmetric monoidal functor, indeed a hypergraph functor.  
\end{thm}

\begin{proof}
By \cite[Thm.\ 4.1]{fong2015} it suffices to check that $\theta \maps C \Rightarrow D$ is indeed a natural transformation and furthermore a \emph{monoidal} natural transformation.
\end{proof}

With more work one can extend the natural transformation $v$ to a monoidal natural transformation between the 2-functors $\mathsf{C} \maps \Fin\Set \to \Cat$ and $\mathsf{D} \maps \Fin\Set \to \Cat$.  By \cite[Thm.\ 2.5]{baez-courser-vasilakopoulou2022}, this gives a symmetric monoidal double functor from $\mathbb{O}\mathbf{pen}(\mathsf{StockFlow})$ to $\mathbb{O}\mathbf{pen}(\mathsf{Dynam})$.  However, we do not need this yet in our code.

\subsection{Full-fledged stock-flow diagrams}
\label{subsec:full-schema}

In Section~\ref{subsec:category} we defined a simple category of stock-flow diagrams, called $\StockFlow$.  Stock-flow diagrams of this type capture two main features of the diagrams used by practitioners: (1) flows between stocks  and (2) links that represent the dependency of flow rates on the values of particular stocks. However, the stock-flow diagrams used in epidemiological modeling have additional useful features.  Our ``full-fledged'' stock-flow diagrams include auxiliary variables, sum variables, and partial flows.   

\emph{Auxiliary variables} are quantities on which flow functions can depend. An auxiliary variable is linked to stocks and other auxiliary variables and is equipped with an arbitrary continuous function of the values of stocks and variables to which it is linked.  In Figure \ref{fig:ExStockAndFlow}, ``Fractional Prevalence", ``Force of Infection", and ``Infection" are all examples of auxiliary variables. Auxiliary variables are important to practitioners for several reasons. First, they simplify model specification because they are reusable: instead of computing each flow rate directly as a function of stocks, we can often compute them more simply with the help of auxiliary variables. Many flow rates can depend on a single auxiliary variable. Second, they often represent quantities that are of interest to stakeholders; representing these quantities explicitly make them easier to track throughout a simulation, such as for comparison with empirical data. Third, they are practical for the communication and revision of models. Auxiliary variables give a meaningful decomposition of the flow functions, and changing a single auxiliary variable automatically revises all the flow functions that depend on it, which eliminates the need to revise all these flow functions separately.

While not explicitly distinguished in current stock-flow modeling packages, flow functions often rely on a special case of auxiliary variables called \emph{sum variables}.  Such a variable equals the sum of the values of some subset of the stocks.  In epidemiology, this frequently corresponds to the size of a population or sub-population. For example, ``Total Population" in Figure \ref{fig:ExStockAndFlow} is a sum variable. In general, a sum variable may link to only a subset of stocks. Sum variables can be seen as a particular type of auxiliary variable in which the function merely sums the values of the stocks to which it is linked---and thanks to this fact, we do not need to label sum variables with functions.

Finally, in the simple stock-flow diagrams described earlier, each flow must have an upstream stock and a downstream stock. However, practitioners often use diagrams including \emph{partial flows}, which may have only an upstream stock or only a downstream stock.  These represent the creation or the destruction of some resource, and are commonly used to represent open populations.

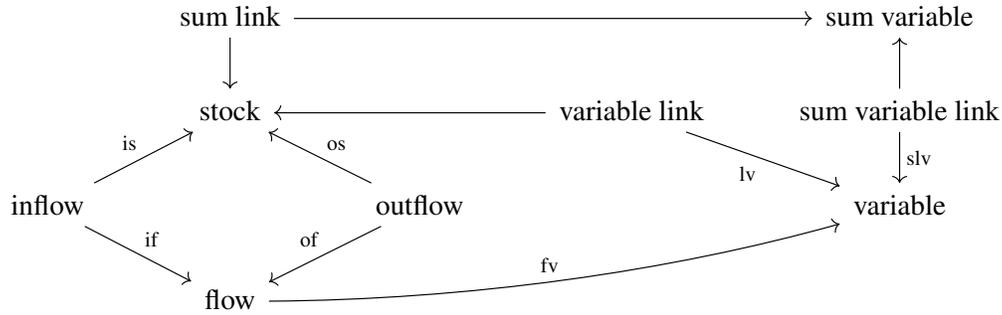
\begin{figure}[ht]
    \centering
    \begin{tikzcd}
    	& \text{sum link} &&& \text{sum variable} \\
    	& \text{stock} && \text{variable link} & \text{sum variable link} \\
    	\text{inflow} && \text{outflow} && \text{variable} \\
    	& \text{flow}
    	\arrow["\textrm{slv}", from=2-5, to=3-5]
    	\arrow["\textrm{is}", from=3-1, to=2-2]
    	\arrow["\textrm{if}", from=3-1, to=4-2]
    	\arrow[swap,"\textrm{of}", from=3-3, to=4-2]
    	\arrow[swap,"\textrm{os}", from=3-3, to=2-2]
    	\arrow["\textrm{fv}", curve={height=12pt}, from=4-2, to=3-5]
    	\arrow[from=2-5, to=1-5]
    	\arrow[from=1-2, to=2-2]
    	\arrow[from=1-2, to=1-5]
    	\arrow[swap,"\textrm{lv}", from=2-4, to=3-5]
    	\arrow[from=2-4, to=2-2]
    \end{tikzcd}
    \caption{The free category on this diagram, called $\H_\ff$, is used to define full-fledged stock-flow diagrams.  We have named only some of the arrows here.}
    \label{fig:full-fledged}
\end{figure}

Figure~\ref{fig:full-fledged} presents the category $\H_\ff$ used to define full-fledged stock-flow diagrams.   A \textbf{full-fledged stock-flow diagram} is a pair $(F,\phi)$ consisting of:
\begin{itemize}
\item a functor $F \maps \H_\ff \to \Fin\Set$ such that the functions $F(\textrm{if})$ and $F(\textrm{of})$ are injective;
\item 
a continuous function $\phi_v \maps \R^{F(\textrm{lv})^{-1}(v)} \times \R^{F(\textrm{slv})^{-1}(v)} \to \R$ for each $v \in F(\variable)$, called an \define{auxiliary function}.
\end{itemize}
The elements of $F(\variable)$ are called \define{auxiliary variables}.   The idea is that in a full-fledged stock-flow diagram each flow has its rate equal to some auxiliary variable.  Each auxiliary variable can depend on any finite multiset of sum variables and stocks, and each sum variable can depend on any finite multiset of stocks.

Given an inflow $f \in F(\textrm{inflow})$, we say that the stock $F(\textrm{is})(f)$ is the \define{upstream} stock of the flow $F(\textrm{if})(f)$.   Similarly, given an outflow  $f \in F(\textrm{outflow})$, the stock $F(\textrm{os})(f)$ is the \define{downstream} stock of the flow $F(\textrm{of})(f)$.  The injectivity of $F(\textrm{if})$ and $F(\textrm{of})$ ensure that each flow has at most one upstream stock  and at most one downstream stock.   Flows having an upstream stock but not a downstream stock or vice versa are called \define{partial flows}.  

Following the ideas of Section~\ref{subsec:category}, we can define a category $\StockFlow_\ff$ of full-fledged stock-flow diagrams.  It is useful to glue together such diagrams not only along stocks but also along sum variables and sum links, for example to keep track of the total population in an epidemiological model.  We can still do this using decorated cospans if we introduce the category $\Fin\Set^\G$, where $\G$ is the free category on this diagram:
    \[
    \begin{tikzcd}
        \stock
        &
        \textrm{ sum link}
        \arrow[l, ""]
        \arrow[r, ""] 
        & \textrm{sum variable.}
    \end{tikzcd}
    \]
There is an evident inclusion functor $\iota \maps \G \to \H_f$, so any functor $F \in \Fin\Set^{\H_f}$ restricts to a functor $F \circ \iota \in \Fin\Set^\G$, and we define an \define{open} full-fledged stock-flow diagram  to be a full-fledged stock-flow diagram $(F,\phi)$ equipped with a cospan
\[
\begin{tikzcd}
& F \circ \iota & \\
X\ar[ur,"i"] && Y \ar[ul,swap,"o"]
\end{tikzcd}
\]
where $X,Y \in \Fin\Set^\G$.   With this adjustment we can define a category $\Open(\StockFlow_\ff)$ of open full-fledged stock-flow diagrams following the ideas in \cref{subsec:open}.  Most importantly, in analogy to \cref{thm:semantics}, there is a functor
\[   v \maps \Open(\StockFlow_\ff) \to \Open(\Dynam) \]
providing a semantics for open full-fledged stock-flow diagrams.   We have implemented full-fledged stock-flow diagrams and this semantics in our Julia package StockFlow (below)---but for simplicity, \cref{sec:implementation} only discusses the simpler stock-flow diagrams treated in Sections \ref{subsec:category}--\ref{subsec:opendynam}.

\section{Implementing stock-flow diagrams in AlgebraicJulia}
\label{sec:implementation}

Existing tools for building stock-flow diagrams and simulating the systems they represent suffer from several limitations.  In \cref{sec:epidemiology}, we singled out two: an absence of compositionality, which it makes it difficult to build complex models in an intelligible manner, and a blurring of the distinction between syntax and semantics, which inhibits the reusability and interoperability of stock-flow diagrams in different contexts. In \cref{sec:math}, we addressed both these problems at the mathematical level, the first by constructing a category of open stock-flow diagrams, and the second by constructing a functor from this category into a category of open dynamical systems, whose morphisms describe systems of differential equations.  We now describe our implementation of these mathematical structures as new software for System Dynamics modeling.

Our software, available at \url{https://github.com/AlgebraicJulia/StockFlow.jl} as the open source package StockFlow, is implemented using AlgebraicJulia \cite{AlgebraicJulia}, a family of packages for applied category theory written in the Julia programming language \cite{julia2017}. The AlgebraicJulia ecosystem consists of Catlab, which implements many standard abstractions in category theory, and a collection of packages which apply these abstractions to specific domains of science and engineering. Most relevant to this article are AlgebraicDynamics \cite{libkind-baas-patterson-fairbanks2021}, implementing open dynamical systems based on ordinary and delay differential equations, and AlgebraicPetri \cite{libkind-baas-halter-patterson-fairbanks2022}, implementing Petri nets with rates and their ODE semantics.

Existing capabilities within AlgebraicJulia, based on general category-theoretic abstractions, enable us to give an economical implementation of stock-flow diagrams. The combinatorial essence of stock-flow diagrams---what we called primitive stock-flow diagrams in Section \ref{subsec:category}---are set-valued functors on a certain category $\H$, or $\H$-sets. Such structures are encompassed by the paradigm of categorical databases, for which Catlab has extensive support \cite{patterson-lynch-fairbanks2021}. Subject to one caveat, stock-flow diagrams---including the flow functions---can also be implemented as categorical databases. In this way, stock-flow diagrams become combinatorial data structures that can be manipulated algorithmically through high-level operations such as limits and colimits.  Currently, Catlab has better support for structured cospans than decorated cospans.  The latter have some theoretical advantages, as mentioned in \cref{subsec:open}, but luckily the two formalisms are equivalent for the tasks carried out here \cite{baez-courser-vasilakopoulou2022}.  Thus, at present we use structured cospans to implement open stock-flow diagrams in StockFlow.  We elaborate on this in the following subsections.

\subsection{Stock-flow diagrams as categorical databases}
\label{subsec:database}

In \cref{subsec:category} we defined a primitive stock-flow diagram to be a finite $\H$-set for a certain category $\H$, called the \define{schema} for these diagrams. In Catlab, we  present this schema as:
\begin{minted}{julia}
@present SchPrimitiveStockFlow(FreeSchema) begin
  (Stock, Flow, Link)::Ob
  (up, down)::Hom(Flow, Stock)
  src::Hom(Link, Stock)
  tgt::Hom(Link, Flow)
end
\end{minted}
To define stock-flow diagrams, we need to add a data attribute for the flow functions. In general, data attributes \cite{patterson-lynch-fairbanks2021} are a practical necessity and the main feature that distinguishes categorical databases from the standard mathematical notion of a $\C$-set, meaning a functor from $\C$ to $\Set$. In this case, we extend the schema with an attribute type and a data attribute:
\begin{minted}{julia}
@present SchStockFlow <: SchPrimitiveStockFlow begin
  FlowFunc::AttrType
  flow::Attr(Flow, FlowFunc)
end
\end{minted}

Having defined the schema, we can generate a Julia data type for stock-flow diagrams with this single line of code:
\begin{minted}{julia}
@acset_type StockFlow(SchStockFlow, index=[:up, :down, :src, :tgt])
\end{minted}
where the indices are generated for morphisms in the schema to enable fast traversal of stock-flow diagrams. A stock-flow diagram will then have the Julia type \texttt{StockFlow\{Function\}}, where \texttt{Function} is the built-in type for functions in Julia. We see that there is a gap between the mathematical definition of stock-flow diagrams and the present implementation: the domains of the flow functions should be constrained by the links, but this constraint is not yet expressible in the data model supported by Catlab. In practice this is not a major obstacle to using stock-flow diagrams, but it could motivate future work toward increasing the expressivity of database schemas and instances in Catlab.

The full-fledged stock flow diagrams described in \cref{subsec:full-schema} are implemented similarly.

\subsection{Composition using structured cospans}
\label{subsec:compose}

While the mathematics described in \cref{sec:math} uses decorated cospans, we can also describe open stock-flow diagrams using structured cospans \cite{baezcourser2020,baez-courser-vasilakopoulou2022}.  A structured cospan is a diagram of the form
\[
\begin{tikzcd}
&  X & \\
L(A) \ar[ur,"i"] && L(B) \ar[ul,swap,"o"]
\end{tikzcd}
\]
in some category $\X$, where $A,B$ are objects in some other category $\A$, and $L \maps \A \to \X$ is a functor.  When $L$ is a left adjoint, we can equivalently think of a structured cospan as a diagram
\[
\begin{tikzcd}
&  R(X) & \\
A \ar[ur,"i"] && B \ar[ul,swap,"o"]
\end{tikzcd}
\]
where $R$ is the right adjoint of $L$.  

For example, we can take $\A = \Fin\Set$, take $\X = \Fin\Set^\H$, and take $R \maps \Fin\Set^\H \to \Fin\Set$ to be the functor sending any primitive stock-flow diagram to its set of stocks.   In this case, a structured cospan amounts to an \define{open primitive stock-flow diagram}, that is a primitive stock-flow diagram $F \maps \H \to \Fin\Set$ together with functions
\[
\begin{tikzcd}
&  F(\stock) & \\
A \ar[ur,"i"] && B. \ar[ul,swap,"o"]
\end{tikzcd}
\]
With more work we can define a structured cospan category equivalent to $\Open(\StockFlow)$.   
The advantage of this change in viewpoint is that Catlab already provides a generic framework for working with structured cospans and multicospans---and it implements the composition operations available in both the hypergraph category of structured cospans and the operad algebra of structured multicospans.
In addition, Catlab includes a concrete instantiation of structured cospans for systems defined by attributed $\C$-sets.  This makes it possible, for a broad class of systems, to use structured cospans in just a few lines of code.  This is the approach taken in StockFlow and also its companion package AlgebraicPetri. In order to support the implementation of full-fledged stock-flow diagrams, we expanded the implementation of structured multicospans in Catlab so that the feet in a multicospan can be  arbitrary $\C$-sets as opposed to merely finite sets.

\subsection{Composing epidemiological models: an example}

\begin{figure}[ht]
    \centering
    \includegraphics[width=0.8\linewidth]{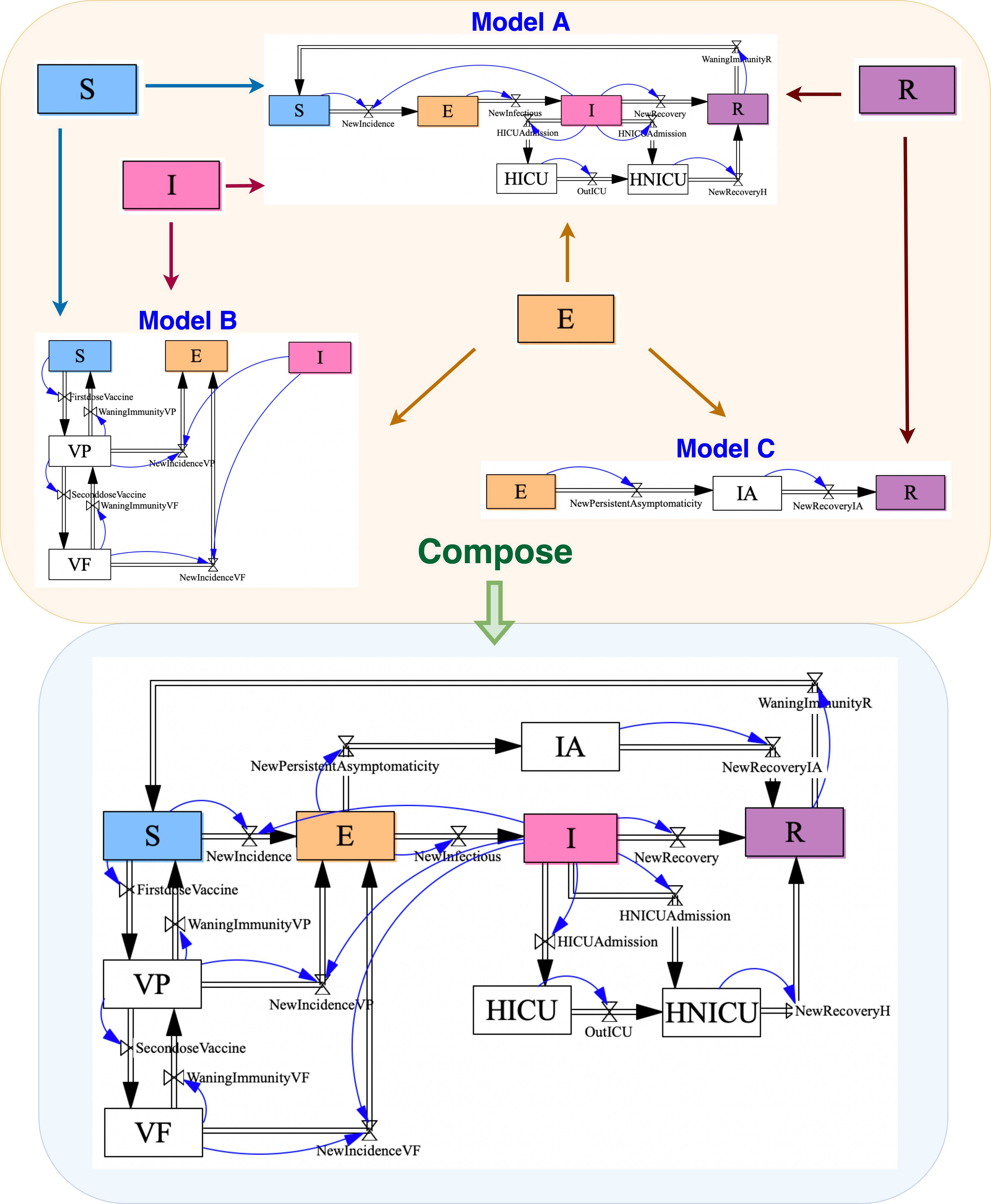}
    \caption{Example of composing a COVID-19 model from three smaller models}
    \label{fig:ExCovid}
\end{figure}

The Julia package StockFlow implements both the open stock-flow diagrams of Section~\ref{subsec:open} and the full-fledged open stock-flow diagrams of Section~\ref{subsec:full-schema}.  We now illustrate the use of this package by constructing a simplified version of a COVID-19 model that has has been employed during the pandemic for daily reporting and planning throughout the Province of Saskatchewan, and for weekly reporting by the Public Health Agency of Canada for each of Canada's ten provinces, as well as by First Nations and Inuit Health for reporting to provincial groupings of First Nations Reserves.

We build this simplified model as the composite of three component models: (A) a model of the natural history of infection, pathogen transmission, and hospitalization, (B) a model of vaccination, and (C) a model of the natural history of infection among asymptomatic or oligosymptomatic individuals.  To exhibit the ideas with a minimum of complexity, we use the simpler open stock-flow diagrams discussed in Sections \ref{subsec:category}-\ref{subsec:semantics}, not the full-fledged ones.  

The top of Figure \ref{fig:ExCovid} shows the open stock-flow diagrams for three models. Although the flow functions are omitted from the figure, they are defined in the Jupyter notebook implementing this example.\footnote{Readers interested in the code for this example can refer to \url{https://github.com/AlgebraicJulia/StockFlow.jl/blob/master/examples/primitive_schema_examples/Covid19_composition_model_in_paper.ipynb} on the GitHub repository for StockFlow.}  Model (A) is the SEIRH (Susceptible-Exposed-Symptomatic Infectious-Recovered-Hospitalized) model, which simulates the disease transmission from, course of infection amongst, and hospitalization of symptomatically infected individuals. The stocks labelled ``HICU" and ``HNICU" represent the populations of hospitalized ICU and non-ICU patients, respectively. Model (B) characterizes vaccination-related dynamics. The stock ``VP" represents individuals who are partially protected via vaccination, due to having been administered only a first dose or to waning of previously full vaccine-induced immunity. In contrast, the stock ``VF" represents individuals who are fully vaccinated by virtue of having received two or more doses of the vaccine. Notably, neither partially or fully vaccinated individuals are considered fully protected from infection. Thus, there are flows from stock ``VP" and ``VF" to ``E" that represent new infection of vaccinated individuals. Model (C) characterizes the natural history of infection in individuals who are persistently asymptomatic.  The stock ``IA" indicates the infected individuals without any symptoms. 

\begin{figure}[t]
    \centering
    \includegraphics[width=0.3\linewidth]{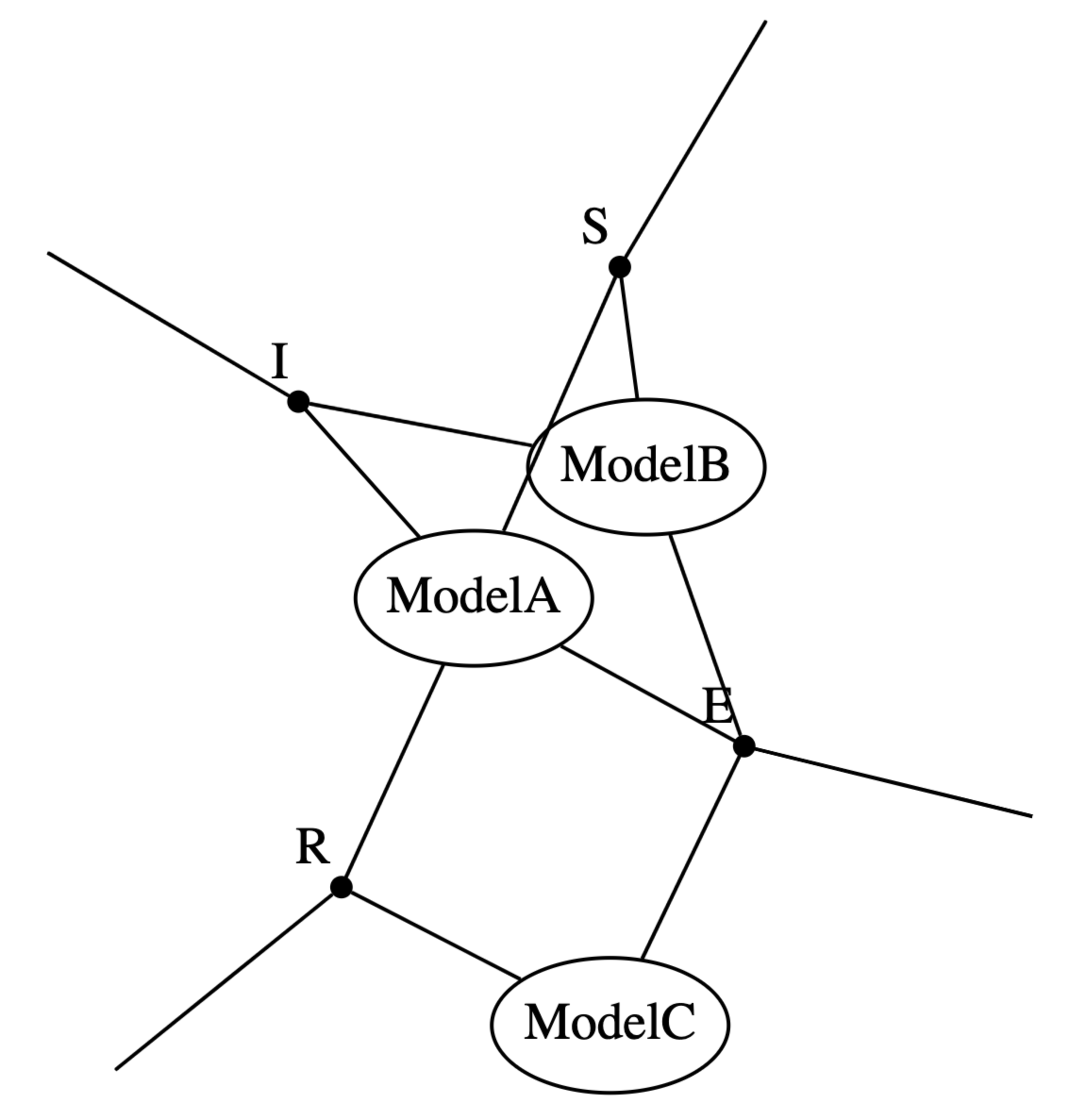}
    \caption{The undirected wiring diagram representing composing structured multicospans}
    \label{fig:ExCovid_compose}
\end{figure}

\begin{figure}[t]
    \centering
    \includegraphics[width=0.5\linewidth]{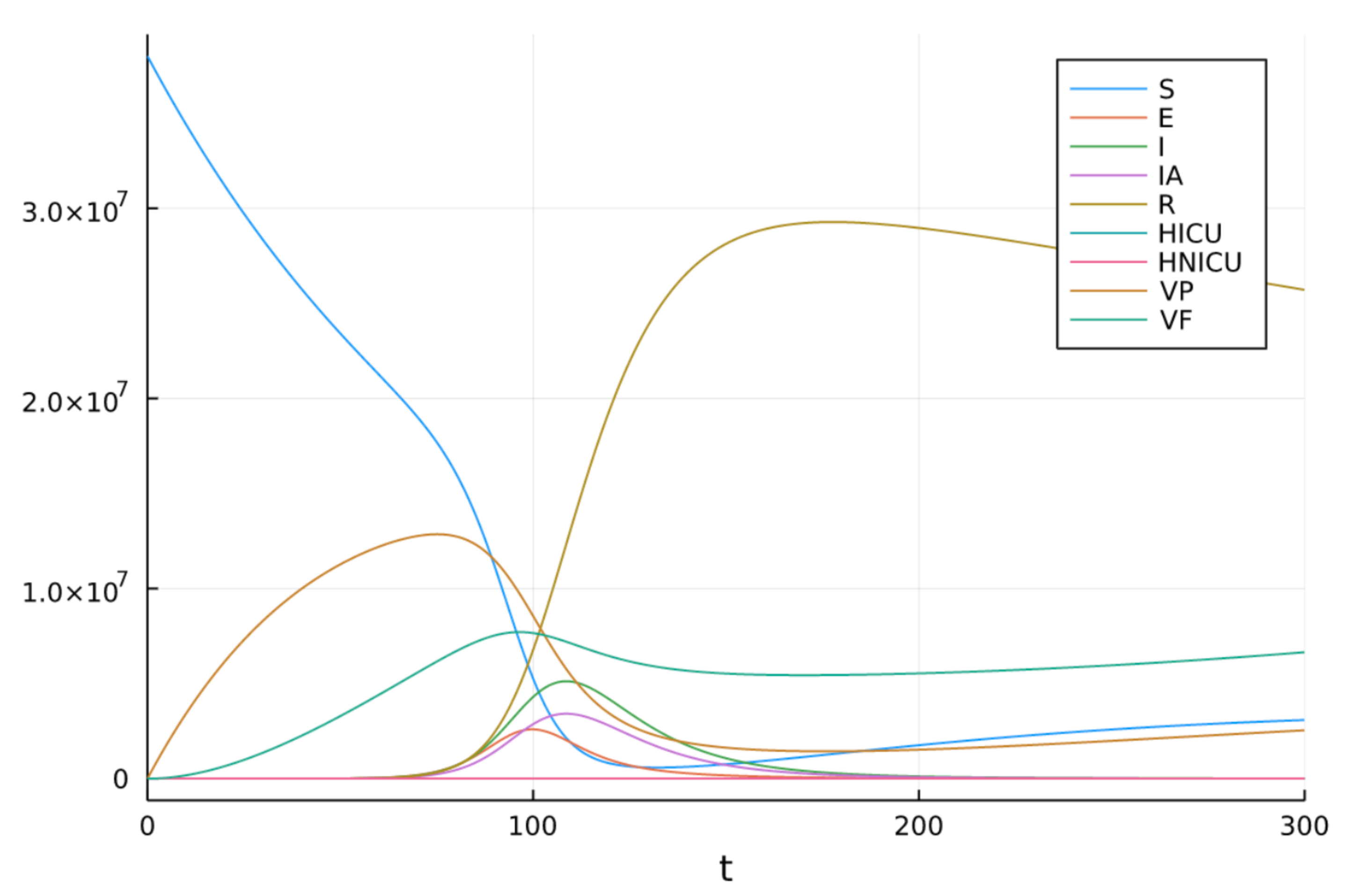}
    \caption{A simulation of the composite COVID-19 model}
    \label{fig:ExCovid_results}
\end{figure}

The ODEs finally generated from the composite stock-flow diagram are as follows:
\begin{align*}
&\dot{S}=\frac{R}{t_w}+\frac{V_P}{t_w}-\frac{\beta SI}{N} - r_v S & &\dot{E}=\frac{\beta SI}{N}+\frac{\beta (1-e_p) I V_P}{N}+\frac{\beta (1-e_f) I V_F}{N}-r_{ia}E-r_{i}E \\
&\dot{I}=r_{i}E-\frac{I}{t_r} & &\dot{R}=\frac{(1-f_H)I}{t_r}+\frac{I_A}{t_r}+\frac{H_{\NICU}}{t_H}-\frac{R}{t_w}\\
&\dot{I_A}=r_{ia}E-\frac{I_A}{t_r} & &\dot{V_F}=r_v V_P - \frac{V_F}{t_w}-\frac{\beta (1-e_f) I V_F}{N} \\
&\dot{H}_{\ICU}=\frac{f_H f_{\ICU} I}{t_r} - \frac{H_{\ICU}}{t_{\ICU}} & &\dot{V_P}=r_v S + \frac{V_F}{t_w} - \frac{V_P}{t_w} - r_v V_P - \frac{\beta (1-e_p) I V_P}{N}\\
& &
&\dot{H}_{\NICU}=\frac{H_{\ICU}}{t_{\ICU}}+\frac{f_H (1-f_{\ICU}) I}{t_r} - \frac{H_{\NICU}}{t_H}
\end{align*}
where for simplicity we use $1/t_r$ to stand for the rate at which infected individuals proceed to the next stage (stocks R, HICU or HNICU), and assume this is also the rate at which asymptomatic infected individuals go to the next stage (stock R). A plot of a solution of these equations is shown in Figure \ref{fig:ExCovid_results}. In our software, the initial values and values of parameters are defined separately from the stock-flow diagram. This design enables the users to flexibly define and explore multiple scenarios involving the same dynamical system, in a manner similar to some existing stock-flow modeling packages. For example, the parameter values used in Figure \ref{fig:ExCovid_results} are from Canada's population. We can efficiently run this model on other populations (e.g., the United States) by changing these parameter values.\footnote{Readers interested in the code for this example can refer to \url{https://github.com/AlgebraicJulia/StockFlow.jl/blob/master/examples/primitive_schema_examples/Covid19_composition_model_in_paper.ipynb} on the GitHub repository for StockFlow.}

This particular COVID-19 model simplifies the structure and assumptions of the model used in practice. Our example omits features such as characterization of active case-finding, diagnosis and reporting, mortality, and transmission by asymptomatic/oligosymptomatic individuals, because the simplified stock-flow diagrams do not support auxiliary variables, sum variables, and partial flows. However, our StockFlow package implements the full-fledged stock-flow diagrams defined in Section~\ref{subsec:full-schema}, and hence enables the application of these additional features.

\subsection{Future work}

Three lines of work are underway to extend the work described here: extending the Julia application programming interface (API), constructing a graphical user interface, and training modelers to use StockFlow.

For the first, key priorities include supporting within-diagram constants in the diagrams and allowing auxiliary variables to depend on other auxiliary variables in an acyclic fashion.  Approaches are also being explored to allow hierarchical composition of diagrams and ensure consistency of the functions governing flows, in the sense of dimensional analysis.

Second, we aim to help modelers use the software without needing to know category theory.  Thus, building atop the API, we are currently constructing a declarative, real-time graphical user interface (GUI) for collaboratively constructing, manipulating, composing and packaging stock-flow diagrams. 

Third, we are training both students and professional modelers in the use of Stockflow, and will ramp up these efforts soon, once the GUI achieves sufficient functionality to offer practical utility.   This will both build a user base and provide useful feedback as to how epidemiological modelers interact with the software.

\paragraph{Acknowledgments}

Author Osgood wishes to express his appreciation to SHA, PHAC and FNIH for support of varying elements of this work, to SYK, and to NSERC for support via the Discovery Grants program (RGPIN 2017-04647).  Patterson acknowledges support from AFOSR (Award FA9550-20-1-0348).  Baez and Libkind thank the Topos Institute for their support.

\bibliographystyle{eptcs}
\bibliography{references}

\end{document}